\documentclass[aps,pra,twocolumn,superscriptaddress,longbibliography]{revtex4-1}
\usepackage{amssymb,amsmath,amsthm}
\usepackage{graphicx}
\usepackage{bm,bbm}
\usepackage[bookmarks=false]{hyperref}
\hypersetup{colorlinks=true,citecolor=blue,linkcolor=red,%
urlcolor=blue,pdfstartview=FitH,bookmarksopen=true}

\usepackage{enumitem}
\usepackage{algorithm}
\usepackage{algcompatible}

\newcommand{\ket}[1]{\mbox{$ | #1 \rangle $}}
\newcommand{\bra}[1]{\mbox{$ \langle #1 | $}}
\newcommand{\braket}[2]{\mbox{$ \langle #1 | #2 \rangle $}}
\newcommand{\tr}{\mathrm{tr}}
\newcommand{\cP}{{\cal P}}
\newcommand{\caZ}{\mathcal{Z}}

\newtheoremstyle{note}
  {\topsep/2}              	
  {\topsep/2}            	
  {}                        
  {\parindent}             	
  {\itshape}                
  {.}                       
  {5pt plus 1pt minus 1pt}  
  {}

\newtheorem{theorem}{Theorem}
\newtheorem{lemma}{Lemma}

\theoremstyle{definition}

\theoremstyle{remark}

\begin{document}
\title{Efficient verification of Dicke states}

\author{Ye-Chao Liu}
\thanks{These authors contributed equally to this work.}
\affiliation{Key Laboratory of Advanced Optoelectronic Quantum Architecture and Measurement,
Ministry of Education and School of Physics, Beijing Institute of Technology, Beijing 100081, China}

\author{Xiao-Dong Yu}
\thanks{These authors contributed equally to this work.}
\affiliation{Naturwissenschaftlich-Technische Fakult\"at, Universit\"at Siegen,
Walter-Flex-Str. 3, D-57068 Siegen, Germany}

\author{Jiangwei Shang}
\email{jiangwei.shang@bit.edu.cn}
\affiliation{Key Laboratory of Advanced Optoelectronic Quantum Architecture and Measurement,
Ministry of Education and School of Physics, Beijing Institute of Technology, Beijing 100081, China}
\affiliation{State Key Laboratory of Surface Physics, Fudan University, Shanghai 200433, China}

\author{Huangjun Zhu}
\email{zhuhuangjun@fudan.edu.cn}
\affiliation{Department of Physics and Center for Field Theory and Particle Physics, Fudan University, Shanghai 200433, China}
\affiliation{State Key Laboratory of Surface Physics, Fudan University, Shanghai 200433, China}
\affiliation{Institute for Nanoelectronic Devices and Quantum Computing, Fudan University, Shanghai 200433, China}
\affiliation{Collaborative Innovation Center of Advanced Microstructures, Nanjing 210093, China}

\author{Xiangdong Zhang}
\email{zhangxd@bit.edu.cn}
\affiliation{Key Laboratory of Advanced Optoelectronic Quantum Architecture and Measurement,
Ministry of Education and School of Physics, Beijing Institute of Technology, Beijing 100081, China}

\date{\today}
%

\begin{abstract}
  Among various multipartite entangled states, Dicke states stand out
  because their entanglement is maximally persistent and robust under particle losses.
  Although much attention has been attracted for their potential
  applications in quantum information processing and foundational studies,
  the characterization of Dicke states remains as a challenging task in experiments.
  Here, we propose efficient and practical protocols for verifying arbitrary
  $n$-qubit Dicke states in both adaptive and nonadaptive ways.
  Our protocols require only two distinct settings based on Pauli measurements
  besides permutations of the qubits. To achieve infidelity $\epsilon$ and confidence
  level $1-\delta$, the total number of tests required is only
  $O(n\epsilon^{-1}\ln\delta^{-1})$.
  This performance is exponentially  more efficient than all previous protocols based on
  local measurements, including quantum state tomography and direct fidelity estimation,
  and is comparable to the best global strategy.
  Our protocols are readily applicable with current experimental techniques and
  are able to verify Dicke states of hundreds of qubits.
\end{abstract}

\maketitle
%

\section{Introduction}
Multipartite quantum states with different types of entanglement are of pivotal
interest in various quantum information processing tasks as well as foundational studies.
Efficient and reliable characterization of these states  plays a crucial role in various
applications.
The standard approach is to fully reconstruct the density
matrix by quantum state tomography \cite{QSE}. However, tomography is both time consuming
and computationally hard due to the exponentially increasing number of parameters
to be reconstructed \cite{Nature438.643,APG_Shang}. Thus, a lot of efforts have
been devoted to searching for non-tomographic methods. Along this research
line there are, for instance, direct entanglement detection
\cite{PRL94.060501, Guehne.Toth2009, Dimic.Dakic2018, arXiv:1809.05455}, direct fidelity estimation (DFE) \cite{PRL106.230501}, self-testing \cite{Mayers2004, Coladangelo2017},
as well as quantum state verification
\cite{PRA96.062321,PRL120.170502,PRX8.021060,arXiv:1806.05565,arXiv:1901.09856,
Li.etal2019,Wang.Hayashi2019,Zhu2019,arXiv:1909.01943}. The latter one aims at
devising efficient protocols for verifying the target states by employing local measurements.
Up to now, efficient (or even optimal) verification protocols for bipartite pure states
have been proposed using both nonadaptive \cite{PRL120.170502,Zhu2019} and
adaptive measurements \cite{arXiv:1901.09856,Li.etal2019,Wang.Hayashi2019}.
Some of these protocols have also been implemented in experiments very recently
\cite{Zhang.etal2019}.
For multipartite states, efficient protocols are known only when the
states admit a stabilizer description, e.g., graph and
hypergraph states \cite{PRL120.170502,PRA96.062321,PRX8.021060,arXiv:1806.05565}.

However, most multipartite states do not admit a stabilizer description,
among which Dicke states \cite{Dicke1954} stand out
particularly as their entanglement is maximally persistent and robust
under particle losses \cite{PRL86.910,PRA78.060301}. Such states are key
resources in various tasks in  quantum information processing, such as multiparty quantum
communication and quantum metrology
\cite{PRL96.100502,PRL98.063604,PRA59.156,PRL103.020503,PRL103.020504,RMP.90.035005}.
In general, an $n$-qubit Dicke state with $k$ excitations is defined as
\begin{equation}\label{eq:Dicke}
  \ket{D_n^k}=\frac1{\sqrt{C_n^k}}\sum_l \cP_l\left\{\ket{1}^{\otimes k}\otimes\ket{0}^{\otimes (n-k)}\right\}\!,
\end{equation}
where $\sum_l \cP_l\{\cdot\}$ denotes the sum over all possible
permutations, and $C_n^k\equiv\binom{n}{k}$ is the binomial coefficient.
When $k=1$, Dicke states are also known as $W$ states \cite{Nature438.643},
\begin{equation}\label{eq:Wn}
  \ket{W_n}=\frac1{\sqrt{n}}(\ket{10\dots0}+\ket{01\dots0}+\cdots+\ket{00\dots1})\,.
\end{equation}

First investigated by Dicke in 1954 for describing light emission from a cloud
of atoms \cite{Dicke1954}, the preparation and characterization of Dicke states
have drawn a lot of theoretical and experimental interest.
Dicke states are relatively easy to generate in experiments
\cite{Nature438.643,PRL98.063604}, for instance Dicke states with up to six
photons have been observed  in photonic systems \cite{PRL103.020503,PRL103.020504}.
Very recently, Dicke states with more than $10\,000$ spin-1 atoms have been successfully
demonstrated in a rubidium condensate \cite{PNAS115.6381}.
In addition, tomography \cite{PRL105.250403,NJP14.105001}, entanglement
characterization
\cite{Guehne.etal2007b,JOSAB24.275,NJP11.083002,PRA83.040301,PRA88.012305,JPA46.385304},
and self-testing \cite{Supic.etal2018,Fadel2017} of
Dicke states can be simplified because of their permutation symmetry.
However, it is quite challenging to verify Dicke states of large quantum systems, and the resource overhead increases exponentially with the number of excitations $k$ even with the best protocols known so far \cite{PRL106.230501}.

In this work, we propose efficient and practical
protocols for verifying arbitrary $n$-qubit Dicke states, including $W$ states, using
both adaptive and nonadaptive measurements. These protocols require only two
distinct measurement settings if permutations of qubits can be realized, and in total
$O(n\epsilon^{-1}\ln\delta^{-1})$ tests suffice to achieve
infidelity $\epsilon$ and confidence level $1-\delta$.
They are exponentially  more efficient than all known strategies based on local measurements,
including tomography, DFE \cite{PRL106.230501}, and self-testing \cite{Supic.etal2018,Fadel2017}; moreover, they are comparable to the best strategy based on
entangling measurements. Our protocols can easily be realized using current technologies
and are able to verify Dicke states of hundreds of qubits.
Moreover, we introduce a general method for constructing nonadaptive verification
protocols from adaptive protocols, which can be applied to the verification of various
other quantum states.

\section{Quantum state verification}
Consider a device that is supposed to produce the target state $\ket{\psi}$,
but may in practice produce $\sigma_1,\sigma_2,\dots,\sigma_N$ in $N$ runs.
In the ideal scenario, we have the promise that either $\sigma_i=\ket{\psi}\bra{\psi}$
for all $i$ or $\bra{\psi}\sigma_i\ket{\psi}\leq1-\epsilon$ for all $i$.
Then the task is to determine which is the case with the worst-case failure probability $\delta$.

In practice, we are interested in two-outcome measurements of the form
$\{\Omega_j,\openone-\Omega_j\}$, where $\Omega_j$ corresponds to passing the test.
A verification protocol takes on the general form
\begin{equation}\label{eq:general}
  \Omega=\sum_{j=1}^m\mu_j \Omega_j\,,
\end{equation}
where $\{\mu_1,\mu_2,\dots,\mu_m\}$ forms a probability distribution.
Here, we require that the target state $\ket{\psi}$ always passes the test, i.e.,
$\Omega_j\ket{\psi}=\ket{\psi}$ for all $\Omega_j$.
Then in the bad case $\bra{\psi}\sigma_i\ket{\psi}\leq1-\epsilon$, the maximal
probability that $\sigma_i$ can pass the test is \cite{PRL120.170502,arXiv:1909.01943}
\begin{equation}
  \max_{\langle\psi|\sigma|\psi\rangle\leq 1-\epsilon}\tr(\Omega\sigma)=1-[1-\lambda_2(\Omega)]\epsilon=1-\nu(\Omega)\epsilon\,,
\end{equation}
where $\lambda_2(\Omega)$ is the second largest eigenvalue of $\Omega$, and
$\nu(\Omega):=1-\lambda_2(\Omega)$ denotes the spectral gap from the maximal eigenvalue.

After $N$ runs, $\sigma$ in the bad case can pass the test with probability at
most $[1-\nu(\Omega)\epsilon]^{N}$.
To achieve confidence level $1-\delta$, i.e.,
$[1-\nu(\Omega)\epsilon]^N\leq\delta$, $N$ needs to satisfy
\cite{PRL120.170502}
\begin{equation}
  N\geq\frac{\ln\delta^{-1}}{\ln\bigl\{[1-\nu(\Omega)\epsilon]^{-1}\bigr\}}\approx\frac1{\nu(\Omega)}\epsilon^{-1}\ln\delta^{-1}\,.
\end{equation}
Therefore, the optimal protocol is obtained by maximizing the spectral gap $\nu(\Omega)$.
If there is no restriction on the accessible measurements, the optimal strategy is simply $\{\ket{\psi}\bra{\psi},\openone-\ket{\psi}\bra{\psi}\}$, so that $\Omega=\ket{\psi}\bra{\psi}$, $\nu(\Omega)=1$, and $N\approx\epsilon^{-1}\ln\delta^{-1}$.
However, it is difficult, if not simply impossible, to realize in experiments
when $\ket{\psi}$ is entangled.
Thus, it is more meaningful to devise efficient strategies based on local measurements only.

\section{Verification of $W$ states}
Besides the permutation symmetry, $\ket{W_n}$
has another important property: if we preform a Pauli-$Z$
measurement on any one of the $n$ subsystems, then the other
subsystems would collapse to either
$\ket{0}^{\otimes(n-1)}$ or $\ket{W_{n-1}}$ depending on whether the outcome is
1 (corresponding to eigenvalue $-1$) or 0 (eigenvalue 1). If we perform $Z$
measurements on all but two qubits, say $i$ and $j$, then outcome 1 can appear
at most once (otherwise, the original state cannot be $\ket{W_n}$). 
If outcome 1 appears, then the reduced state of parties $i$ and
$j$ is $\ket{00}$, which can be verified easily by $Z$ measurements on the two
parties; if outcome 1 does not appear, then the reduced state of parties $i$
and $j$ is $\ket{W_2}=\frac1{\sqrt{2}}(\ket{01}+\ket{10})$, which is nothing but a
Bell state.
This state can be verified optimally using the following protocol \cite{PRL120.170502, JPA39.14427, arXiv:0810.3381, Zhu2019}
\begin{equation}\label{eq:Bell}
  \Omega_{\mathrm{Bell}}=\frac1{3}\left[(XX)^{+}+(YY)^{+}+(ZZ)^{-}\right],
\end{equation}
where $X,Y,Z$ are the Pauli operators. Here the symbols $\pm$ in the superscripts
indicate the projectors onto the eigenspaces with eigenvalues $\pm1$.
See Appendix~A for more details on the  verification of a Bell state. In this
way, we can construct a test for $\ket{W_n}$ for each pair $i$ and $j$. By
randomizing the choices of $i$ and $j$ we can devise a verification protocol.

\begin{figure}[t]
  \includegraphics[width=.9\columnwidth]{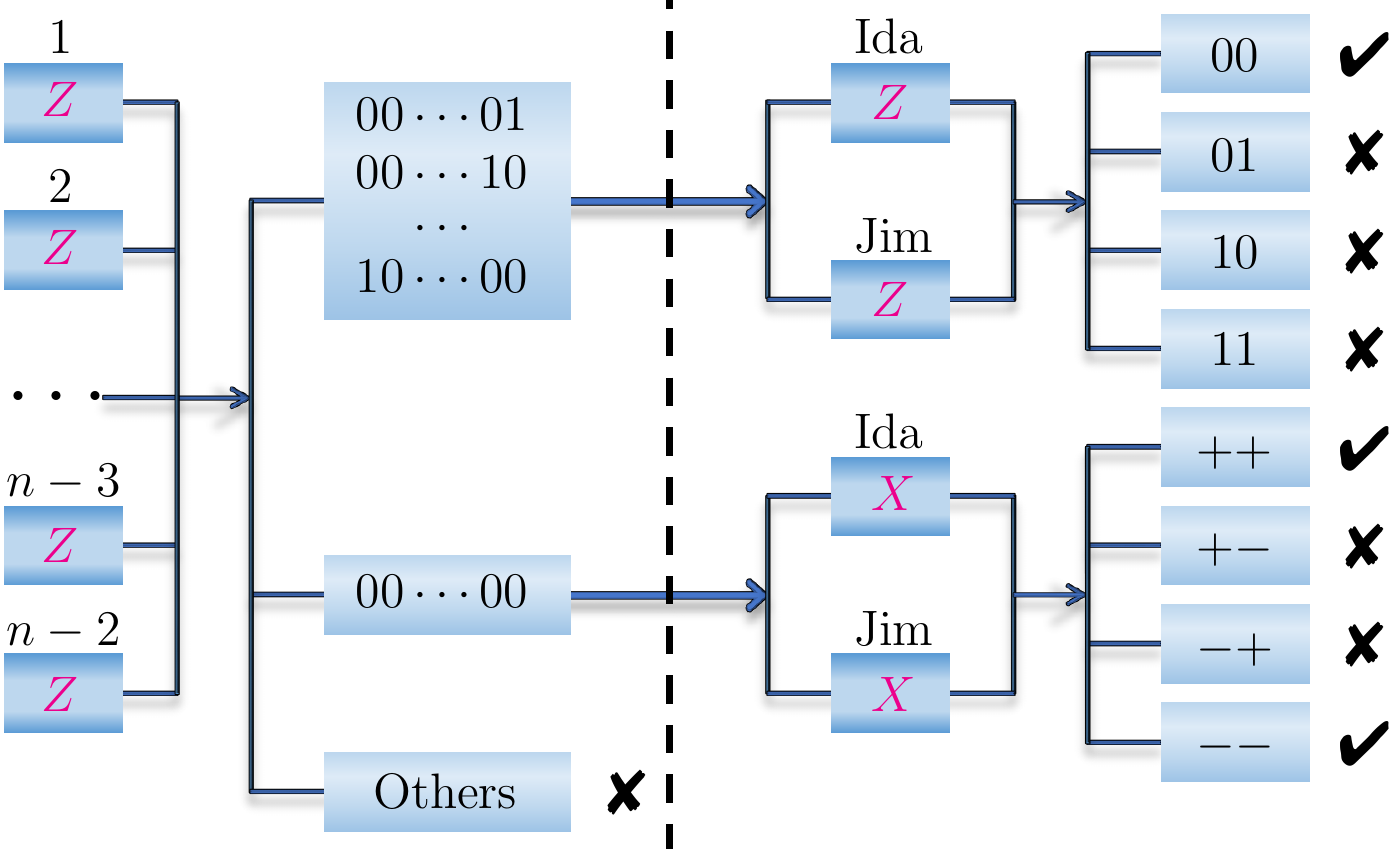}
  \caption{\label{fig:schemeW}
  Schematic view of the adaptive  protocol for verifying $\ket{W_n}$ as in Theorem~\ref{theo:W}.
  The dashed vertical line indicates that the protocol is two-step adaptive.
  For any two qubits $i$ and $j$ (Ida and Jim) chosen \textit{a priori},
  the measurement outcomes of the other $n-2$ qubits determine which measurements
  on them to perform.
  }
\end{figure}

It turns out that the tests based on $(YY)$ and $(ZZ)$ measurements in
Eq.~\eqref{eq:Bell} can be dropped out if randomization is considered.
The resulting protocol is illustrated in Fig.~\ref{fig:schemeW}, and its efficiency
is guaranteed by the following theorem, which is proved in Appendix~B.
\begin{theorem}\label{theo:W}
	$\ket{W_n}$ can be verified efficiently using the strategy
	\begin{eqnarray}\label{eq:OmegaWT}
	\Omega_W&=&\frac1{C_n^2}\sum_{i<j}\Omega_{i,j}^{\rightarrow}
	\end{eqnarray}
	with
	\begin{eqnarray}\label{eq:OmegaW}
	\Omega_{i,j}^{\rightarrow}=\bar{\cal Z}_{i,j}^1(Z_i^{+}Z_j^{+})
	+\bar{\cal Z}_{i,j}^0(XX)_{i,j}^{+}\,,
	\end{eqnarray}
	where the notation $\bar{\cal Z}_{i,j}^k$ means that  $k$ excitations
	are detected when we perform $Z$ measurements on all qubits except for $i$ and $j$.
	The spectral gap is $\nu(\Omega_{W})=\frac{1}{3}$ when $n=3$ and
	\begin{equation}
	\nu(\Omega_W)=\frac1{n-1}\quad \mbox{for}\,\, n\ge4\,.
	\end{equation}
\end{theorem}

The test $\Omega_{i,j}^{\rightarrow}$ in Eq.~\eqref{eq:OmegaWT} can be realized using
adaptive measurements with two distinct measurement settings:
$n-2$ parties except for parties $i$ and $j$ perform $Z$  measurements,
then parties $i$ and $j$ perform either $Z$ or $X$ measurements
depending on whether an excitation is detected or not in the first stage.
The strategy $\Omega_W$ is composed of $C_n^2=\frac1{2}n(n-1)$ tests with probability
$2/[n(n-1)]$ each. Since all these tests can be turned into each other by permuting the qubits,
our protocol can be realized using only two measurement settings if permutations of qubits
can be realized.

Before proceeding further, we show that Theorem~\ref{theo:W} inspires
an efficient nonadaptive protocol, although the verification efficiency would
deteriorate by a factor of 2.  The basic idea is to replace the adaptive test
$\Omega_{i,j}^{\rightarrow}$ with two nonadaptive tests, performed with equal
probability. In one test, all parties perform $Z$ measurements, and the test is
passed if  excitation is detected once. In the other test, parties $i$ and $j$
perform $X$ measurements, and the other $n-2$ parties perform $Z$ measurements;
the test is passed if one excitation is detected for $Z$ measurements, or no
excitation is detected and the outcomes for parties $i$ and $j$ coincide. The
respective test projectors read
\begin{align}\label{eq:Omega01}
  \caZ^1&=\bar{\cal Z}_{i,j}^1(Z_i^{+}Z_j^{+})
  +\bar{\cal Z}_{i,j}^0(ZZ)^-_{i,j}\,,\\
  \Omega_{i,j}&=\bar{\cal Z}_{i,j}^0(XX)_{i,j}^{+}
  +\bar{\cal Z}_{i,j}^1(\openone\openone)_{i,j}\,.
\end{align}
Here $\caZ^1$ can also be expressed as $\caZ^1=\sum_{u\in
B_{n,1}}\ket{u}\bra{u}$ with $B_{n,1}$ being the
set of strings in $\{0,1\}^n$ with Hamming weight $1$. Note that $\caZ^1$ is
independent of $i,j$, unlike $\Omega_{ij}$.
The resulting verification operator reads
\begin{equation}\label{eq:OmegaWN}
\tilde\Omega_W=\frac1{2C_n^2}\sum_{i<j}\left(\caZ^1+\Omega_{i,j}\right)
=\frac{1}{2}\caZ^1+\frac1{2C_n^2}\sum_{i<j}\Omega_{i,j},
\end{equation}
and the spectral gap satisfies
\begin{equation}\label{eq:efficiency}
\nu(\tilde\Omega_W)\geq \nu\!\left(\frac{1}{2}\Omega_W+\frac{1}{2}\openone^{\otimes n}\right)
\ge\frac{1}{2}\nu(\Omega_W)\,.
\end{equation}
This bound is actually saturated when $n\geq 4$ (see the proof in
Appendix~C); in the case $n=3$, direct calculation shows that
$\nu(\tilde\Omega_{W_3})=\frac{3}{4}\nu(\Omega_{W_3})$.
Hence, the verification efficiency of $\tilde\Omega_W$ is
worse than that of the adaptive protocol $\Omega_W$ by a factor of at most 2.

When $n=3$ for example, we have
\begin{eqnarray}
  \Omega_{W_3}&=&\frac1{3}\Bigl[Z_3^{-}(Z_2^{+}Z_1^{+})+Z_3^{+}(XX)_{2,1}^{+}+Z_2^{-}(Z_3^{+}Z_1^{+})\nonumber\\
	    &+&Z_2^{+}(XX)_{3,1}^{+}+Z_1^{-}(Z_3^{+}Z_2^{+})+Z_1^{+}(XX)_{3,2}^{+}\Bigr]
\end{eqnarray}
for the adaptive protocol.
It is easy to verify that $\lambda_2(\Omega_{W_3})=\frac{2}{3}$ and $\nu(\Omega_{W_3})=\frac1{3}$.
So the number of tests required to verify $\ket{W_3}$ within infidelity $\epsilon$ and
confidence $1-\delta$ is $N\approx3\epsilon^{-1}\ln\delta^{-1}$. For the nonadaptive protocol
$\tilde{\Omega}_{W_3}$, we have $\nu(\tilde{\Omega}_{W_3})=\frac1{4}$, so the number of tests
required is $N\approx4\epsilon^{-1}\ln\delta^{-1}$. These results are corroborated by
numerical simulations in which we choose the worst noise in the eigenspace corresponding to
the second largest eigenvalue and get
$N\approx3.0031(\pm 0.0169)\epsilon^{-1}\ln\delta^{-1}$ for the adaptive protocol and
$N\approx3.9806(\pm 0.0109)\epsilon^{-1}\ln\delta^{-1}$ for the nonadaptive one.
Similarly, we get $N\approx7.0306(\pm 0.0188)\epsilon^{-1}\ln\delta^{-1}$ (adaptive) and
$N\approx14.0621(\pm 0.0262)\epsilon^{-1}\ln\delta^{-1}$ (nonadaptive) for $\ket{W_8}$.
More details on the simulated experiments can be found in Appendix~D.

\section{Verification of Dicke states}
Our protocols for verifying $W$ states can be naturally generalized to
arbitrary $n$-qubit Dicke states $\ket{D_n^k}$. Since $\ket{D_n^{n-1}}$ is equivalent to
$\ket{D_n^{1}}=\ket{W_n}$ under a local unitary transformation, we can assume $2\leq k\leq n-2$
and $n\geq 4 $ without loss of generality. For any pair of parties $i$ and $j$, we can
construct a test as follows (see Fig.~$4$ in Appendix~E for an illustration).
First we perform $Z$ measurements on $n-2$ parties other than parties $i$ and $j$.
If the outcomes have $k$ or $k-2$ excitations,
then we perform $(ZZ)$ measurements on qubits $i$ and $j$ and the test is passed
if the total number of excitations is $k$; if the outcomes have $k-1$ excitations,
then we perform $(XX)$ measurements and the test is passed if the two outcomes for
parties $i$ and $j$ coincide.
By randomizing the choice of the pair $i,j$
we can construct a verification protocol that is composed of $n(n-1)/2$ tests.
The efficiency of this protocol is guaranteed by the following theorem,
which is proved in Appendix~E.
\begin{theorem}\label{theo:D}
$\ket{D_n^k}$ can be verified efficiently using the strategy
\begin{eqnarray}\label{eq:OmegaDT}
  \Omega_D&=&\frac1{C_n^2}\sum_{i<j}\Omega_{i,j}^{\rightarrow}
\end{eqnarray}
with
\begin{align}\label{eq:OmegaD}
  \Omega_{i,j}^{\rightarrow}=\bar{\cal Z}_{i,j}^k(Z_i^{+}Z_j^{+})
  +\bar{\cal Z}_{i,j}^{k-2}(Z_i^{-}Z_j^{-})
  +\bar{\cal Z}_{i,j}^{k-1}(XX)_{i,j}^{+}\,.
\end{align}
The spectral gap is
\begin{equation}
  \nu(\Omega_D)=\frac1{n-1}\quad \mbox{for}\,\, n\ge4\,.
\end{equation}
\end{theorem}

Several remarks are in order. First, when $k=1$, the second term in Eq.~\eqref{eq:OmegaD}
drops out and we get back Eq.~\eqref{eq:OmegaW} as expected.
Second, although we need to consider  three different cases in constructing
$\Omega_{i,j}^{\rightarrow}$, only two distinct measurement settings are required,
which is the same as that for $W$ states.
Last, the spectral gap $\nu(\Omega_D)$ is independent of $k$ and is
the same as that for $W$ states. Therefore, all $n$-qubit Dicke states can be verified
using the same experimental setup and with the same efficiency.
To achieve infidelity $\epsilon$ and confidence $1-\delta$, the total number of tests
required is only $N\approx(n-1)\epsilon^{-1}\ln\delta^{-1}$, so our protocol is able to verify
Dicke states of hundreds of qubits.

Similar to the case of $W$ states, Theorem~\ref{theo:D} also inspires an
efficient nonadaptive protocol.
The basic idea is to replace the adaptive test $\Omega_{i,j}^{\rightarrow}$ with
two nonadaptive tests as characterized by the two test projectors
\begin{align}\label{eq:OmegaD01}
&\caZ^k=\bar{\cal Z}_{i,j}^k(Z_i^{+}Z_j^{+})
+\bar{\cal Z}_{i,j}^{k-2} (Z_i^{-}Z_j^{-}) +\bar{\cal
Z}_{i,j}^{k-1}(ZZ)^-_{i,j}\,,\\
 &\Omega_{i,j}=\bar{\cal Z}_{i,j}^{k-1}(XX)_{i,j}^{+}
+\bar{\cal Z}_{i,j}^k(\openone\openone)_{i,j}+\bar{\cal
Z}_{i,j}^{k-2}(\openone\openone)_{i,j}\,.
\end{align}
Here $\caZ^k$ can also be expressed as
$\caZ^k=\sum_{u\in B_{n,k}}\ket{u}\bra{u}$ with $B_{n,k}$ being the set of strings in
$\{0,1\}^n$ with Hamming weight $k$.
The resulting verification operator reads
\begin{equation}\label{eq:OmegaDN}
\tilde\Omega_D=\frac1{2C_n^2}\sum_{i<j}\left(\caZ^k+\Omega_{i,j}\right)
=\frac{1}{2}\caZ^k+\frac1{2C_n^2}\sum_{i<j}\Omega_{i,j},
\end{equation}
and the spectral gap satisfies
\begin{equation}\label{eq:efficiencyD}
\nu(\tilde\Omega_D)\geq \nu\!\left(\frac{1}{2}\Omega_D+\frac{1}{2}\openone^{\otimes n}\right)
\ge\frac{1}{2}\nu(\Omega_D)\,.
\end{equation}
This bound is actually saturated given the assumption $n\geq 4$ (see the proof in Appendix~F).
Hence, the verification efficiency of $\tilde\Omega_D$ is
worse than that of the adaptive protocol $\Omega_D$  by a factor of 2.

Take $\ket{D_4^2}$ as an example.
The second largest eigenvalue and spectral gap of $\Omega_{D_4^2}$ (see Appendix~G for
an explicit expression) read
$\lambda_2(\Omega_{D_4^2})=\frac{2}{3}$ and
$\nu(\Omega_{D_4^2})=\frac1{3}$.
So the number of tests required to verify $\ket{D_4^2}$ within infidelity $\epsilon$
and confidence $1-\delta$ is $N\approx3\epsilon^{-1}\ln\delta^{-1}$.
For the nonadaptive protocol $\tilde{\Omega}_{D_4^2}$, we have
$\nu(\tilde{\Omega}_{D_4^2})=\frac1{6}$, so the number of tests required is
$N\approx6\epsilon^{-1}\ln\delta^{-1}$. 
See the numerical confirmations in Appendix~D.

\begin{figure}[t]
  \includegraphics[width=.9\columnwidth]{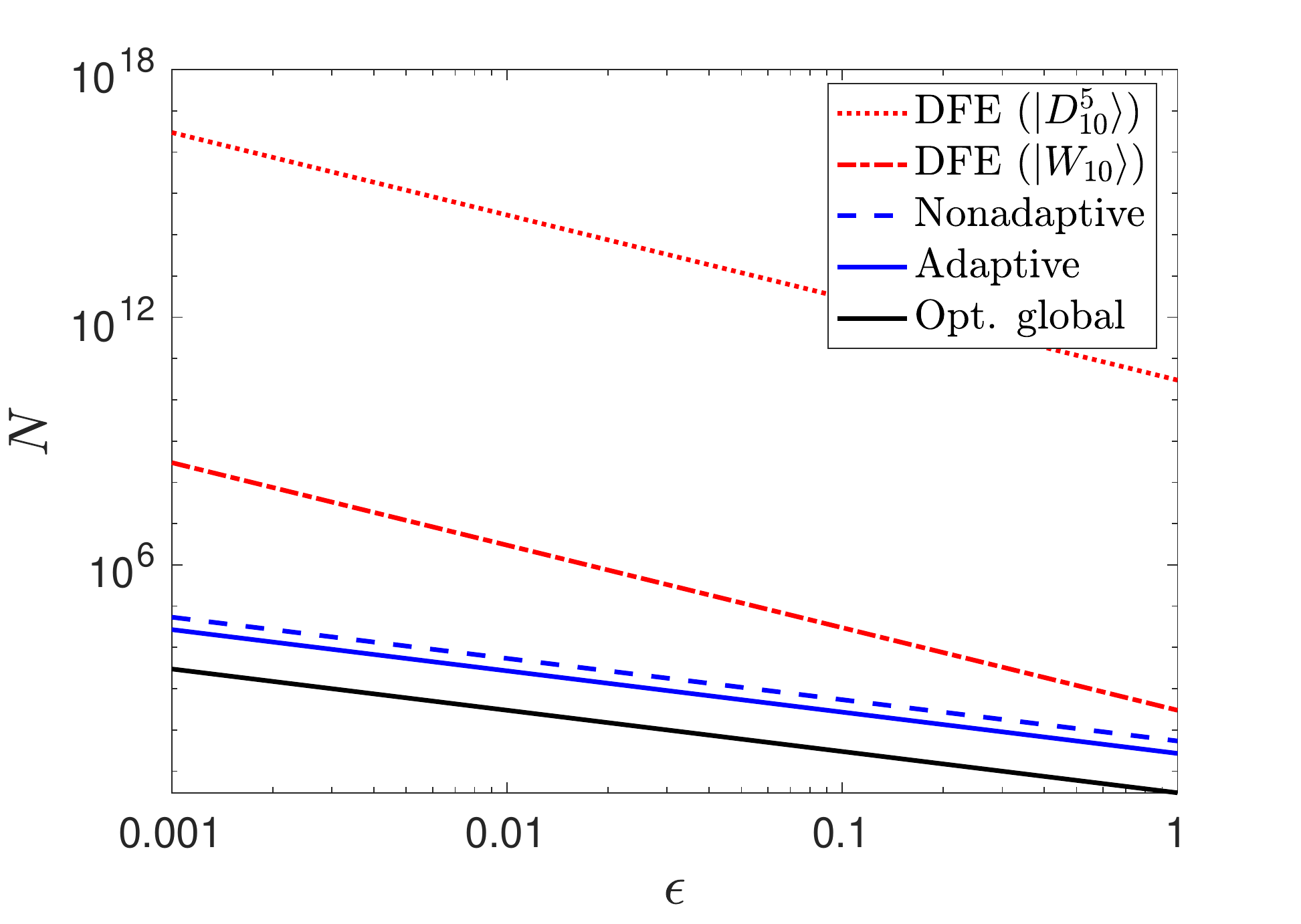}
  \caption{\label{fig:compare}
  Comparison between our efficient adaptive and nonadaptive protocols with the DFE protocol \cite{PRL106.230501} and the optimal global protocol.
  For a given number of qubits $n=10$ and  confidence $1-\delta=0.95$,
  the plot shows the number of tests $N$ required to verify $\ket{W_{10}}$ and $\ket{D_{10}^5}$ within infidelity $\epsilon$
  for each protocol. 
  }
\end{figure}
%

\section{Comparison with other methods}
Here, we compare our adaptive and nonadaptive protocols with
two other non-tomographic methods.
The first one is the protocol of  direct fidelity estimation (DFE)
proposed in Ref.~\cite{PRL106.230501}.
For an $n$-qubit Dicke state with $k$ excitations, this protocol
requires $N\propto O(n^{2k}\epsilon^{-2}\ln\delta^{-1})$ tests,
and the number of measurement settings has the same order of magnitude.
This number increases exponentially with $n$ if
${k\propto n}$, which is the case for the balanced Dicke state with $k=n/2$.
The second one is the optimal global verification protocol with the
entangled verification operator $\Omega=\ket{D_n^k}\bra{D_n^k}$, which requires
$N\approx\epsilon^{-1}\ln\delta^{-1}$ tests.

In Fig.~\ref{fig:compare}, by fixing the number of qubits $n=10$ and the
confidence $1-\delta=0.95$, we plot the number of tests  $N$ required to verify
$\ket{W_{10}}$ and $\ket{D_{10}^5}$ within infidelity $\epsilon$.
As can be seen, our adaptive and nonadaptive protocols are much more efficient than DFE
and are comparable to the best protocol based  on entangling measurements.
In addition, similar to the optimal global protocol, the performances of our protocols are
independent of the number of excitations $k$, while the performance of DFE deteriorates quickly
as $k$ increases and is already impractical for $k=5$ and $\epsilon=0.1$.

\section{Construction of nonadaptive protocols from adaptive protocols}
Inspired by the above results, here we present a general method for converting adaptive
verification protocols to nonadaptive ones at the price of  efficiency.
To this end, we need a notion for characterizing the complexity of an adaptive protocol.
As shown in Fig.~\ref{fig:schemeW}, an adaptive test
is usually composed of a number of branches. The branch number of the test $\Omega_j$,
denoted by $\alpha(\Omega_j)$, is defined as the total number of such branches in realizing
$\Omega_j$, and the branch number of a protocol  is the maximum branch number over all tests.
For example, the branch numbers of the adaptive protocols
$\Omega_W$ and $\Omega_D$ are $2$ and $3$, respectively.
To construct a nonadaptive protocol, we can replace each adaptive test with
a number of nonadaptive tests depending on the branch number,
which sets a lower bound for the efficiency.
More precisely, we have the following theorem.
\begin{theorem}\label{the:adap-nonadap}
  In quantum state verification, an adaptive protocol
  $\Omega=\sum_{j=1}^m\mu_j\Omega_j$ can always be converted to a nonadaptive one
  $\tilde\Omega$ whose spectral gap satisfies
  \begin{equation}\label{eq:adap-nonadap}
    \nu(\tilde{\Omega})\ge\frac{1}{\alpha}\nu(\Omega)\,,
  \end{equation}
  where $\alpha=\max_j\{\alpha(\Omega_j)\}$ is the branch number of the adaptive protocol $\Omega$.
\end{theorem}
\begin{proof}
For simplicity, here we consider a two-step adaptive test of the form (but our
idea is applicable in general)
\begin{equation}
  \Omega_j=\sum_{a=1}^{\alpha_j}M_{a|j}\otimes N_{a|j}\,,
  \label{eq:adaptive}
\end{equation}
where $\alpha_j=\alpha(\Omega_j)$, and $\{M_{a|j}\}_a$ represents a (possibly
incomplete) generalized measurement on subsystem $A$, i.e., $M_{a|j}\ge 0$
and $\sum_{a=1}^{\alpha}M_{a|j}\le\openone_A$, while $N_{a|j}$ represents
a test on subsystem $B$ that depends on the outcome $a$ and  satisfies $0\leq
N_{a|j}\leq \openone_B$. Here both $A$ and $B$ may consist of one or more
subsystems. Based on the adaptive test
$\Omega_j$
we can construct $\alpha_j$ nonadaptive tests
\begin{equation}\label{eq:Ptilde}
  \tilde{\Omega}_{a|j}=M_{a|j}\otimes N_{a|j}+\sum_{b\ne a}M_{b|j}\otimes
  \openone_B\,,
\end{equation}
and the corresponding nonadaptive strategy is
\begin{equation}\label{eq:nonadaptive}
  \tilde{\Omega}=\sum_{j=1}^m\sum_{a=1}^{\alpha_j}\frac{\mu_j}{\alpha_j}\tilde{\Omega}_{a|j}\,,
\end{equation}
which satisfies  $\tilde{\Omega}\ket{\psi}=\ket{\psi}$ whenever
$\Omega\ket{\psi}=\ket{\psi}$.
Now, we have
\begin{equation}\label{eq:effnon}
  \nu(\tilde{\Omega})=\nu\!\left(\frac{1}{\alpha}\Omega+\Omega'\right)
  \ge\frac{1}{\alpha}\nu(\Omega)\,,
\end{equation}
where the inequality follows from
$\Omega'=\sum_{j=1}^m\mu_j\bigl(1-\frac{1}{\alpha_j}\bigr)
\bigl(\sum_{a=1}^{\alpha_j}M_{a|j}\bigr)\otimes\openone_B+
\sum_{j=1}^m\mu_j\bigl(\frac{1}{\alpha_j}
-\frac{1}{\alpha}\bigr)\sum_{a=1}^{\alpha_j}\bigl(M_{a|j}\otimes N_{a|j}\bigr)
\le \bigl(1-\frac{1}{\alpha}\bigr)\openone$.
\end{proof}

According to Theorem~\ref{the:adap-nonadap}, an efficient adaptive protocol can
be converted to an efficient nonadaptive one if the branch number
$\alpha$ is small. This is the case for our adaptive protocols
for verifying $W$ states and Dicke states, in which $\alpha$ equals $2$ and $3$,
respectively. In addition, the adaptive protocols proposed in
Refs.~\cite{arXiv:1901.09856,Li.etal2019,Wang.Hayashi2019} for general
bipartite pure states can be  converted to nonadaptive ones by our method.
Note that in the above construction, we are interested in a general recipe; for
a specific adaptive strategy, sometimes one can construct better nonadaptive
strategies. For instance, if several branches happen to require the same measurement
setting, we can merge these branches into one, which is
the case for the verification of Dicke states.

\section{Conclusions}
Efficient and reliable characterization of quantum states plays a vital role
in almost all quantum information processing tasks as well as foundational studies.
Using both adaptive and nonadaptive approaches,
here we proposed efficient and practical protocols for verifying arbitrary $n$-qubit
Dicke states, including $W$ states. Both adaptive and nonadaptive protocols require
only two distinct settings based on Pauli measurements together with
permutations of the qubits, which is well within the reach of current experimental techniques.
To verify an $n$-qubit Dicke state within infidelity $\epsilon$ and confidence
$1-\delta$, both protocols require only  $O(n\epsilon^{-1}\ln\delta^{-1})$ tests,
which is exponentially  more efficient than all previous protocols based on local measurements.
Thus, our protocols are able to verify Dicke states of hundreds of qubits.
Moreover, we introduced a general method for constructing nonadaptive verification protocols
from adaptive protocols.
Our work opens the possibility of efficiently verifying
many other interesting multipartite states in the future,
even the possibility of developing a general verification
strategy for all states eventually.

\acknowledgments
We are grateful to Otfried G\"uhne, Yun-Guang Han, and Zihao Li for discussions.
This work has been supported by the National Key R\&D Program of China under Grant No.
2017YFA0303800 and the National Natural Science Foundation of China through Grant Nos.
11574031, 61421001, 11805010, and 11875110.
J.S. also acknowledges support by the Beijing Institute of Technology Research Fund Program
for Young Scholars.
X.D.Y. acknowledges support by the DFG and the ERC (Consolidator Grant 683107/TempoQ).

\appendix

\section{Verification of Bell states}\label{app:Bell}
Bell states can be verified optimally using the protocol in Ref.~\cite{PRL120.170502}
(see also  Refs.~\cite{JPA39.14427, arXiv:0810.3381,Zhu2019}).
For the particular Bell state that we consider in this work, i.e.,
$\ket{W_2}=\frac1{\sqrt{2}}(\ket{01}+\ket{10})$, the optimal strategy reads
\begin{equation}
  \Omega_{\mathrm{Bell}}=\frac1{3}\left[(XX)^{+}+(YY)^{+}+(ZZ)^{-}\right],
\end{equation}
which reproduces Eq.~\eqref{eq:Bell} in the main text, and the spectral gap is
$\nu(\Omega_{\mathrm{Bell}})=\frac{2}{3}$.  To verify the Bell state within infidelity
$\epsilon$ and confidence level $1-\delta$, the number of required tests is
$N\approx\frac3{2}\epsilon^{-1}\ln\delta^{-1}$.

As an alternative, one can  modify the optimal protocol by
removing the test based on measurement $(YY)$
\cite{PRL120.170502,Zhu2019}.
Then  the verification  operator reads
\begin{equation}
  \Omega_{W_2}=\frac1{2}\left[(XX)^{+}+(ZZ)^{-}\right],
\end{equation}
and the spectral gap reduces to $\nu(\Omega_{\mathrm{Bell}})=\frac{1}{2}$.
Accordingly, the number of tests increases to $N\approx2\epsilon^{-1}\ln\delta^{-1}$.
This protocol requires only two measurement settings instead of three although the efficiency
is slightly worse. This observation was instrumental in constructing efficient protocols
for verifying $W$ and Dicke states at the beginning of our study.

\section{Proof of Theorem~1}
Theorem~1 is an immediate consequence of the following lemma, which provides
more details
on the verification operator $\Omega_W$.
\begin{lemma}\label{lem:Wverify}
	For $n\ge 3$, the verification operator $\Omega_W$ has five different eigenvalues
	$1,1-\frac1{n-1},\frac1{2}+\frac1{n(n-1)},\frac1{n(n-1)},0$
	with multiplicities  $1, n-1, 1, \frac1{2}n(n-1)-1$, and
	$2^n-\frac{1}{2}(n^2+n)$, respectively.
	When $n= 3$, the second largest eigenvalue of $\Omega_W$ is
	$\lambda_2(\Omega_W)=\frac2{3}$, which is nondegenerate, and the spectral gap is
    $\nu(\Omega_W)=\frac{1}{3}$. When $n\geq 4$, the second largest eigenvalue  is
	$\lambda_2(\Omega_W)=1-\frac1{n-1}$ with multiplicity $n-1$, the spectral gap is
    $\nu(\Omega_W)=\frac{1}{n-1}$, and the corresponding eigenspace is spanned by
	\begin{equation}\label{eq:2ndEig}
	\ket{\phi_{ij}}=\ket{\psi^{-}}_{i,j}\otimes\ket{0}^{\otimes(n-2)}\,,\quad 1\leq i<j\leq n\,.
	\end{equation}
	where $\ket{\psi^-}=\frac{1}{\sqrt{2}}(\ket{01}-\ket{10})$ is the singlet.
\end{lemma}

\begin{proof}
For $n\geq 3$,  recall that $\Omega_W$ is defined as
  \begin{align}\label{eq:decomOmegaW}
    \Omega_W
      =&\frac{1}{C_n^2}\sum_{i<j}\bar{\mathcal{Z}}_{i,j}^1(Z_i^{+}Z_j^{+})
      +\frac{1}{C_n^2}\sum_{i<j}\bar{\mathcal{Z}}_{i,j}^0(XX)_{i,j}^{+}\nonumber\\
      =&\frac{1}{C_n^2}\sum_{i<j}\bar{\mathcal{Z}}_{i,j}^1(Z_i^{+}Z_j^{+})
      +\frac{1}{C_n^2}\sum_{i<j}\bar{\mathcal{Z}}_{i,j}^0
      \otimes\bigl(\ket{\psi^+}\bra{\psi^+}\bigr)_{i,j}\nonumber\\
      &+\frac{1}{C_n^2}\sum_{i<j}\bar{\mathcal{Z}}_{i,j}^0
      \otimes\bigl(\ket{\varphi^+}\bra{\varphi^+}\bigr)_{i,j}\,,
  \end{align}
  where $\ket{\psi^+}=\frac{1}{\sqrt{2}}(\ket{01}+\ket{10})$ and $\ket{\varphi^+}=\frac{1}{\sqrt{2}}(\ket{00}+\ket{11})$ are Bell states.
  Direct calculations show that
  \begin{widetext}
    \begin{align}
      &\frac{1}{C_n^2}\sum_{i<j}\bar{\mathcal{Z}}_{i,j}^1(Z_i^{+}Z_j^{+})
      =\frac{n-2}{n}\mathcal{Z}^1\,,\\
      &\frac{1}{C_n^2}\sum_{i<j}\bar{\mathcal{Z}}_{i,j}^0
      \otimes\bigl(\ket{\psi^+}\bra{\psi^+}\bigr)_{i,j}
      =\frac{1}{2C_n^2}\Big[n\ket{W_n}\bra{W_n}
      +(n-2)\mathcal{Z}^1\Big],\\
      &\frac{1}{C_n^2}\sum_{i<j}\bar{\mathcal{Z}}_{i,j}^0
      \otimes\bigl(\ket{\varphi^+}\bra{\varphi^+}\bigr)_{i,j}
      =\frac{1}{2C_n^2}\Big[C_n^2\ket{\phi_0}\bra{\phi_0}
	+\left(\mathcal{Z}^0
      +\mathcal{Z}^2-\ket{\phi_1}\bra{\phi_1}\right)\!\Big],
    \end{align}
  \end{widetext}
  where
    \begin{align}\label{eq:defphi}
      \ket{\phi_0}&=\,\,{\cal 
      N}\!\left[\sqrt{C_n^2}\ket{D_n^0}+\ket{D_n^2}\right]\!,\\
      \ket{\phi_1}&=\,\,{\cal 
      N}\!\left[\ket{D_n^0}-\sqrt{C_n^2}\ket{D_n^2}\right]\!,
    \end{align}
  with ${\cal N}[\cdot]$ denoting the normalization of the vector inside.
  Note that $\ket{W_n}$ belongs to the support of  $\mathcal{Z}^1$, while $\ket{\phi_0}$
  and $\ket{\phi_1}$ belong to the support of $\mathcal{Z}^0+\mathcal{Z}^2$ and satisfy
  $\braket{\phi_0}{\phi_1}=0$. Therefore, $\Omega_W$ has five different eigenvalues
  $1,1-\frac1{n-1},\frac1{2}+\frac1{n(n-1)},\frac1{n(n-1)},0$
  with multiplicities  1, $n-1$, 1, $\frac1{2}n(n-1)-1$, and
  $2^n-\frac{1}{2}(n^2+n)$, respectively.
  The second largest eigenvalue of $\Omega_W$ is achieved either in the
  support of $\mathcal{Z}^1-\ket{W_n}\bra{W_n}$ or in the eigenvector $\ket{\phi_0}$, that is,
  \begin{equation}
    \lambda_2(\Omega_W)=\max\left\{1-\frac{1}{n-1},\frac{1}{2}+\frac{1}{n(n-1)}\right\}\!,
  \end{equation}
  for all $n\ge 3$.
  Accordingly, the spectral gap from the largest eigenvalue reads
  \begin{equation}
    \nu(\Omega_W)=\min\left\{\frac{1}{n-1},\frac{1}{2}-\frac{1}{n(n-1)}\right\}\!,
  \end{equation}
  for all $n\ge 3$.
  It is easy to see that $\lambda_2(\Omega_W)=\frac{2}{3}$ and $\nu(\Omega_W)=\frac{1}{3}$
  when $n=3$, and the corresponding eigenvector is $\ket{\phi_0}$ in Eq.~\eqref{eq:defphi}.
  When $n\ge 4$, we have
  $\lambda_2(\Omega_W)=1-\frac{1}{n-1}$, $\nu(\Omega_W)=\frac{1}{n-1}$, and the
  corresponding eigenspace coincides with the support of
  $\mathcal{Z}^1-\ket{W_n}\bra{W_n}$, which is an $(n-1)$-dimensional subspace
  spanned by the kets $\ket{\phi_{ij}}$ in Eq.~\eqref{eq:2ndEig}.
\end{proof}

To briefly summarize, Theorem~1 provides an efficient verification protocol for
any $n$-qubit $W$ state $\ket{W_n}$.
In real experiments, the experimenter needs to perform the following procedure in each
run of the verification protocol:
\setlist{rightmargin=1cm}
\begin{itemize}
  \item
    The $n$ parties use shared randomness or classical communication to
    randomly choose any two parties, e.g., $i$ and $j$ (Ida and Jim).
  \item
    All the rest $n-2$ parties perform Pauli-$Z$ measurements, then send
    their outcomes to Ida and Jim.
    \begin{itemize}
      \item
	If the outcome~$1$ does not appear for all the $n-2$ Pauli-$Z$
	measurements, then both Ida and Jim  perform Pauli-$X$ measurements.
	\begin{itemize}
	  \item
	  If the two Pauli-$X$ measurements give the same outcome, then they
	  announce the result ``pass''; otherwise they announce the result
	  ``fail''.
	\end{itemize}
      \item
	If the outcome~$1$ appears exactly once for the $n-2$ Pauli-$Z$
	measurements, then both Ida and Jim perform Pauli-$Z$ measurements.
	\begin{itemize}
	  \item
 If both of them obtain outcome~$0$,
	  then they announce the result ``pass''; otherwise, they announce the
	    result ``fail''.
	\end{itemize}
      \item
	If the outcome~$1$ appears more than once for the $n-2$ Pauli-$Z$
	measurements, then Ida and Jim announce the result ``fail''.
    \end{itemize}
\end{itemize}

\section{Proof of the saturation of the bound in Eq.~\eqref{eq:efficiency}}
In this appendix, we prove the saturation of the bound in  Eq.~\eqref{eq:efficiency} when $n\ge 4$.
In the case $n=3$, direct calculation shows that $\nu(\tilde\Omega_{W_3})=\frac{3}{4}\nu(\Omega_{W_3})$.
See below the lemma.
\begin{lemma}\label{lem:Wefficiency}
    When $n\geq 4$, $\nu(\tilde\Omega_W)=\frac{1}{2}\nu(\Omega_W)$.
\end{lemma}

\begin{proof}
  In the main text, we have already proved that
  $\nu(\tilde\Omega_W)\ge\frac{1}{2}\nu(\Omega_W)$. Hence, to prove the saturation of the bound, we just need to show that
  $\nu(\tilde\Omega_W)\le\frac{1}{2}\nu(\Omega_W)$. Be reminded that
  $\nu(\tilde\Omega_W)$ can be written as
  \begin{equation}
    \nu(\tilde{\Omega}_W):=1-\max_{\braket{\phi}{W_n}=0}\bra{\phi}\tilde{\Omega}_W\ket{\phi}\,.
    \label{eq:vW}
  \end{equation}
  By taking $\ket{\phi}$ to be $\ket{\phi_{ij}}$ defined in
  Eq.~\eqref{eq:2ndEig}, which are orthogonal to $\ket{W_n}$, we get an
  upper bound of $\nu(\tilde{\Omega}_W)$, i.e.,
  \begin{equation}
    \nu(\tilde{\Omega}_W)\le 1-\bra{\phi_{ij}}\tilde{\Omega}_W\ket{\phi_{ij}}
    =\frac{1}{2(n-1)}=\frac{1}{2}\nu(\Omega_W)\,.
    \label{eq:vWupper}
  \end{equation}
This inequality completes the proof.
\end{proof}

\section{Simulated experiments on quantum state verification}
Here we show how to perform simulated experiments on quantum state verification (QSV).
As a demonstration, we use the verification protocols for $W$ states and Dicke states
as characterized by the operators
$\Omega_{W/D}=\frac1{C_n^2}\sum_{i<j}\Omega_{i,j}^{\rightarrow}$.

To set the input state, we add noise to the target state $\ket{\psi}$ such that
\begin{equation}\label{eq:noisyPsi}
  \ket{\psi'}=\sqrt{1-\epsilon}\ket{\psi}+\sqrt{\epsilon}\ket{\tau}\,,
\end{equation}
where the noisy state $\ket{\tau}$ is chosen in the vector space corresponding
to the second largest eigenvalue of $\Omega_{W/D}$.
Other kinds of noise, including random noise, would be easier to detect.
Then, for each input state $\ket{\psi'}$, we perform one of the $C_n^2$ tests
$\Omega_{i,j}^{\rightarrow}$ randomly with probability $1/{C_n^2}$ each.
If \ket{\psi'} passes the test (with probability
$\tr(\Omega_{i,j}^{\rightarrow}\ket{\psi'}\bra{\psi'})$), we continue with the next one.
Otherwise, the verification protocol ends and we record the number of ``pass" instances.
This process is repeated many times, from which we calculate the minimum number of tests
required to achieve a given confidence level $1-\delta$.
Specifically, the simulation procedure can be formulated as in the following algorithm.
\begin{center}
\begin{algorithm}[H]
\caption{Simulated experiments on QSV}\label{alg:sim}
\begin{algorithmic}[1]
\STATEx {\bf Input:} The target state $\ket{\psi}$, the noise \ket{\tau}, the infidelity
$\epsilon$, and the confidence level $1-\delta$.
\STATEx {\bf Objective:} Determine the number of tests $N$ required to
verify $\ket{\psi}$ within infidelity $\epsilon$ and confidence $1-\delta$.	
\STATE {\bf Init:} Set the input state $\ket{\psi'}$ as in
Eq.~\eqref{eq:noisyPsi}.		\STATE{\bf Measure:} Perform one of the
$C_n^2$ tests $\Omega_{i,j}^{\rightarrow}$ on $\ket{\psi'}$ randomly with
probability $1/{C_n^2}$ each.
\STATE{\bf Count:} If $\ket{\psi'}$ passes the test, then repeat step 2.
Otherwise, end the verification procedure and record the number of ``pass" instances $N_i$.
\STATE{\bf Loop:} Repeat steps 2 and 3 above $M$ times and record the $M$ numbers  $N_i$ for $i=1,2,\ldots,M$.
\STATE{\bf Output:}  Arrange $N_i$ in decreasing order and then output $N:=N_{\lfloor\delta M\rfloor}$.
\end{algorithmic}
\end{algorithm}
\end{center}

\begin{figure}[t]
\centering
\includegraphics[width=.9\columnwidth]{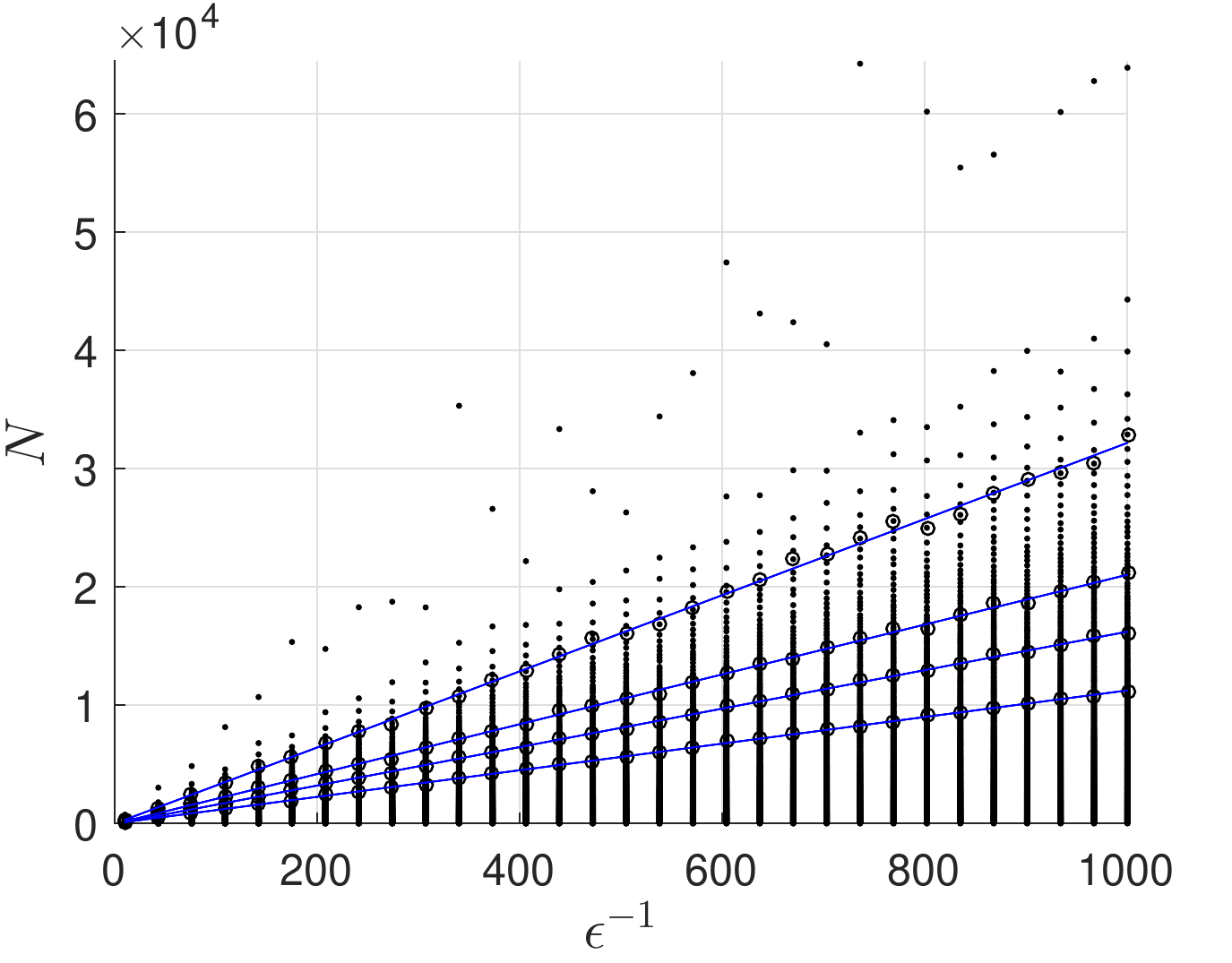}
  \caption{\label{fig:sim_W8}
  Simulation results on the adaptive verification of $\ket{W_8}$.
  The two axes denote the number of tests $N$ and the reciprocal of the infidelity $\epsilon^{-1}$.
  For each $\epsilon$, the simulation is  repeated $M=10000$ times, but only
  500 points are shown in the plot for clarity.
  The open circles denote the minimum number of tests required to  achieve
  infidelity $\epsilon$ and confidence $1-\delta$, and the four blue lines are fitted for
  $\delta=0.01, 0.05, 0.1, 0.2$ (top-down), respectively.
  Here the number of tests can be approximated by the formula
  $N\approx7.0306(\pm 0.0188)\epsilon^{-1}\ln\delta^{-1}$, which is very close
  to the theoretical prediction.
  }
\end{figure}
\setlength{\tabcolsep}{0.85em}
\renewcommand{\arraystretch}{1.2}
\begin{table}
  \caption{\label{tab:sim}
  Simulation results on the verification of $W$ and Dicke states.
  For each state, the verification procedure is repeated $M=10000$ times for both
  the adaptive and nonadaptive protocols. The table shows the fitted values
  of the parameter $\frac{1}{\nu(\Omega)}$ featuring
  in the formula $N\approx\frac{1}{\nu(\Omega)}\epsilon^{-1}\ln\delta^{-1}$.
  The values inside the parentheses are the standard deviations calculated from
  100 different instances of $\delta$ taken uniformly from the interval $0.01$ to $0.2$.
  }
\begin{tabular}{c c c}
\hline\hline
 State            &Adaptive              &Nonadaptive    \\
\hline
 \ket{W_{3}}      &$3.0031(\pm 0.0169)$ &$3.9806(\pm 0.0109)$\\
 \ket{W_{4}}      &$3.0088(\pm 0.0066)$ &$6.0640(\pm 0.0362)$\\
 \ket{W_{5}}      &$3.9847(\pm 0.0076)$ &$7.9825(\pm 0.0379)$\\
 \ket{W_{6}}      &$4.9916(\pm 0.0186)$ &$9.9982(\pm 0.0310)$\\
 \ket{W_{7}}      &$5.9984(\pm 0.0149)$ &$11.9445(\pm 0.0495)$\\
 \ket{W_{8}}      &$7.0306(\pm 0.0188)$ &$14.0621(\pm 0.0262)$\\
 \ket{D_{4}^{2}}  &$2.9931(\pm 0.0058)$ &$5.9411(\pm 0.0181)$\\
 \ket{D_{5}^{2}}  &$4.0064(\pm 0.0193)$ &$7.9722(\pm 0.0245)$\\
 \ket{D_{6}^{2}}  &$4.9981(\pm 0.0110)$ &$10.0371(\pm 0.0233)$\\
 \ket{D_{6}^{3}}  &$4.9654(\pm 0.0148)$ &$10.0147(\pm 0.0173)$\\
 \ket{D_{7}^{2}}  &$6.0131(\pm 0.0118)$ &$11.9302(\pm 0.0240)$\\
 \ket{D_{7}^{3}}  &$5.9610(\pm 0.0111)$ &$11.9553(\pm 0.0366)$\\
 \ket{D_{8}^{2}}  &$6.9554(\pm 0.0378)$ &$14.0045(\pm 0.0345)$\\
 \ket{D_{8}^{4}}  &$6.9554(\pm 0.0361)$ &$14.0669(\pm 0.0289)$\\
\hline\hline
\end{tabular}
\end{table}

As an example, the simulation  results on the verification of $\ket{W_8}$ using the
adaptive protocol are shown in Fig.~\ref{fig:sim_W8}.
By numerical fitting we get the approximation $\frac{1}{\nu(\Omega)}\approx7.0306(\pm 0.0188)$,
which is very close to the theoretical value of 7.
For more simulation results, see Table~\ref{tab:sim}.

\begin{figure}
  \includegraphics[width=.9\columnwidth]{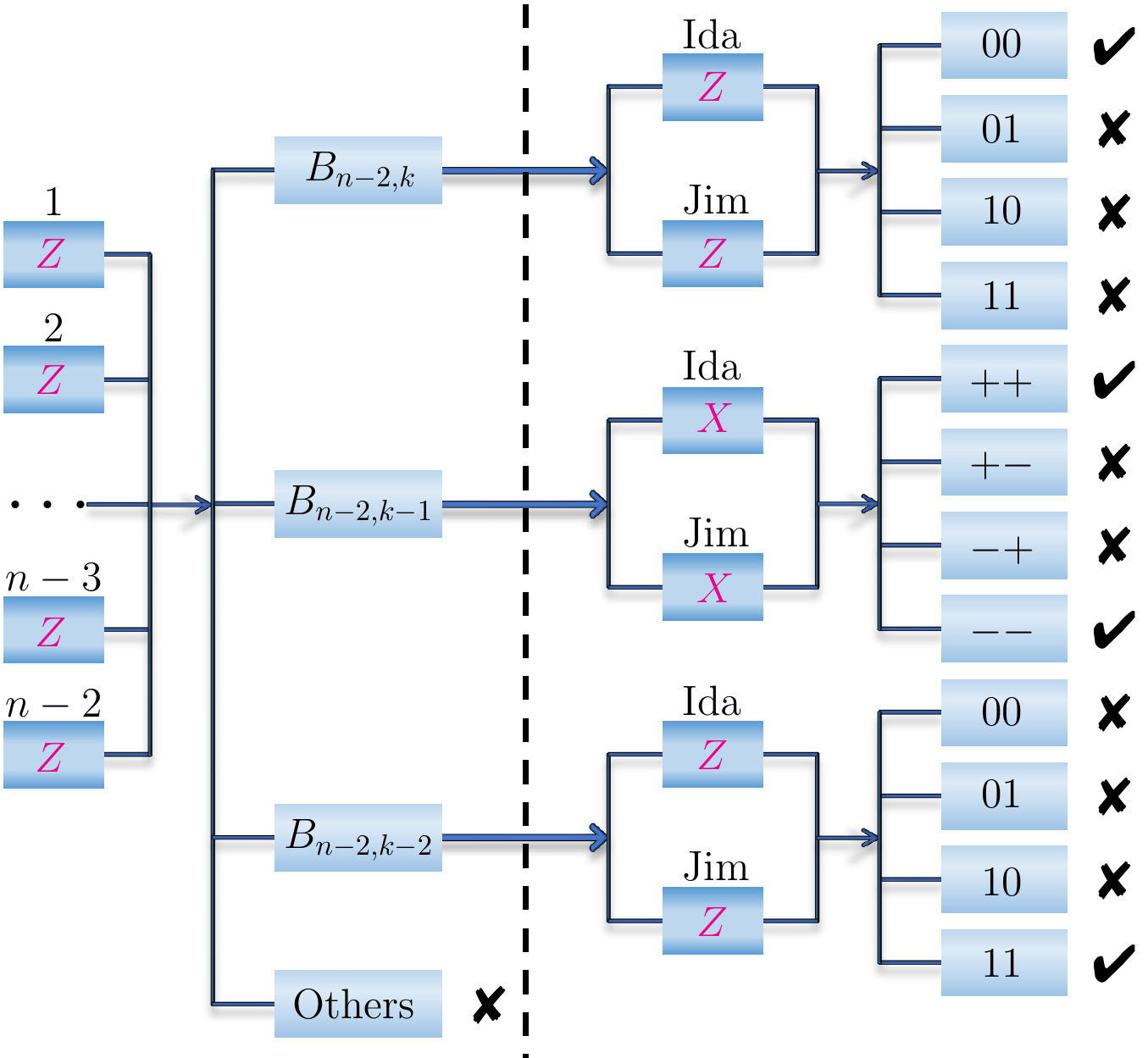}
  \caption{\label{fig:schemeD}
  Schematic view of the adaptive  protocol for verifying the Dicke state
  $\ket{D_{n}^k}$ as in Theorem~\ref{theo:D}.
  The protocol is two-step adaptive as indicated by the dashed vertical line.
  For any two qubits $i$ and $j$ (Ida and Jim) chosen \textit{a priori},
  the outcomes of $Z$ measurements on the other $n-2$ qubits determine which measurements
  on them to perform.
  Recall that $B_{n,k}$ denotes the set of all strings in $\{0,1\}^n$ that have Hamming weight $k$, where 0 and 1 correspond to eigenvalues 1 and $-1$ of $Z$, respectively.
  }
\end{figure}
%

\section{Proof of Theorem~2}
In this Appendix, we prove Theorem~2 and give more details on the adaptive verification
protocol for Dicke states.
First, see Fig.~\ref{fig:schemeD} for a schematic view of this verification protocol.
Theorem~2 is a consequence of the following lemma, which is a generalization of
Lemma~\ref{lem:Wverify}.

\begin{lemma}
	For $n\ge4$, the second largest eigenvalue of $\Omega_D$ in Eq.~\eqref{eq:OmegaDT} is
	$\lambda_2(\Omega_D)=1-\frac1{n-1}$ with multiplicity $n-1$, and the
	corresponding eigenspace is spanned by
	\begin{equation}\label{eq:2ndEigD}
	\ket{\phi_{ij}}=\ket{\psi^{-}}_{i,j}\otimes\ket{D_{n-2}^{k-1}}\,,\quad 1\leq i<j\leq n\,,
	\end{equation}
	where $\ket{\psi^-}=\frac{1}{\sqrt{2}}(\ket{01}-\ket{10})$ is the singlet.
\end{lemma}
\begin{proof}
When $k=1$ or $k=n-1$, the conclusion follows from  Lemma~\ref{lem:Wverify},
so here we can assume $2\leq k\leq n-2$.
Recall that $\Omega_D$ is defined as
\begin{widetext}
\begin{align}\label{eq:decomOmegaD}
      \Omega_D
      =&\frac{1}{C_n^2}\sum_{i<j}\big[\bar{\mathcal{Z}}_{i,j}^{k}(Z_i^{+}Z_j^{+})
      +\bar{\mathcal{Z}}_{i,j}^{k-2}(Z_i^{-}Z_j^{-})\big]
      +\frac{1}{C_n^2}\sum_{i<j}\bar{\mathcal{Z}}_{i,j}^{k-1}(XX)_{i,j}^{+}\nonumber\\
      =&\frac{C_{n-2}^k+C_{n-2}^{k-2}}{C_n^k}\mathcal{Z}^k
      +\frac{1}{C_n^2}\sum_{i<j}\bar{\mathcal{Z}}_{i,j}^{k-1}
      \otimes\bigl(\ket{\psi^+}\bra{\psi^+}\bigr)_{i,j}+\frac{1}{C_n^2}\sum_{i<j}\bar{\mathcal{Z}}_{i,j}^{k-1}
      \otimes\bigl(\ket{\varphi^+}\bra{\varphi^+}\bigr)_{i,j}\nonumber\\
      =&\frac{1}{n(n-1)}\big(M_1+M_2\big)\,,
\end{align}
where
\begin{align}
M_1=&\big[n(n-1)-k(n-k)\big]\caZ^k+\sum_{\substack{u,v\in B_{n,k}\\ u-v\in B_{n,2} }}\ket{u}\bra{v}
=\big[n(n-1)-k(n-k)\big]\caZ^k+\sum_{u,v\in B_{n,k} }A_{u
v}\ket{u}\bra{v}\,,\label{eq:M1}\\
M_2=&\frac{(n-k)(n-k+1)}{2}\sum_{u\in B_{n,k-1} }\ket{u}\bra{u}
+\frac{k(k+1)}{2}\sum_{v\in B_{n,k+1} }\ket{v}\bra{v}
+\sum_{\substack{u\in B_{n,k-1}\\v\in B_{n,k+1}\\u-v\in B_{n,2}} 
}\big(\ket{u}\bra{v}+\ket{v}\bra{u}\big)\,,
\label{eq:M2}
\end{align}
\end{widetext}
and $\ket{\psi^+}=\frac{1}{\sqrt{2}}(\ket{01}+\ket{10})$ and 
$\ket{\varphi^+}=\frac{1}{\sqrt{2}}(\ket{00}+\ket{11})$ are Bell states.
Here $B_{n,k}$ denotes the set of all strings in $\{0,1\}^n$ that have Hamming 
weight $k$,
and the bitwise operation $u-v$ is modulo~2; the coefficient matrix $(A_{uv})$ for $u,v\in B_{n,k}$ happens to be the
adjacency matrix of the Johnson graph $J(n,k)$ \cite{BrouCN89}.
Note that $M_1$ and  $M_2$ are hermitian and have orthogonal supports, so both of them are
positive semidefinite given that $\Omega_D$ is positive semidefinite by construction.

According to Theorem 9.1.2 in Ref.~\cite{BrouCN89}, the distinct eigenvalues of $A$ and
corresponding multiplicities read
\begin{equation}
(k-j)(n-k-j)-j, \quad C_n^j-C_n^{j-1},
\end{equation}
for all $j=0,1,\ldots, \min\{k,n-k\}$,
where it is understood that $C_n^{-1}=0$.
Therefore, the  two largest eigenvalues of $M_1$ read
\begin{equation}\label{eq:M1lambda}
  \begin{aligned}
    \lambda_1(M_1)&=n(n-1)\,, \\
    \lambda_2(M_1)&=n(n-1)-n=n(n-2)\,,
  \end{aligned}
\end{equation}
which have multiplicities 1 and $n-1$, respectively.

Now, we consider $M_2$. Direct calculations show that $M_2$ has an eigenvector
\begin{equation}
  \ket{\phi}={\cal N}\!\left[\sqrt{C_n^{k+1}}\ket{D_n^{k-1}}
  +\sqrt{C_n^{k-1}}\ket{D_n^{k+1}}\right]\!.
  \label{eq:defphiD}
\end{equation}
As $M_2$ is irreducible in the subspace spanned by $\ket{u}$ with $u\in B_{n,k-1}$
or $u\in B_{n,k+1}$, i.e., the graph corresponding to the third term of
$M_2$ in Eq.~\eqref{eq:M2} is connected,  Perron-Frobenius theorem (see
e.g., Chapter~8 in Ref.~\cite{Meyer2000}) implies that the eigenvalue
corresponding to the ket in
Eq.~\eqref{eq:defphiD} is the largest (and nondegenerate) eigenvalue of $M_2$,
which reads
\begin{equation}
  \lambda_1(M_2)=\frac{1}{2}n(n+1)+k(k-n)\,.
  \label{eq:eigM_2}
\end{equation}
In conjunction with Eqs.~\eqref{eq:decomOmegaD} and \eqref{eq:M1lambda}, we can deduce
the second largest eigenvalue and its spectral gap from the largest eigenvalue,
\begin{align}
 &\lambda_2(\Omega_D)=\max\left\{1-\frac{1}{n-1},
 \frac1{2}+\frac{k(k-n)+n}{n(n-1)}\right\}\!,\\
 &\nu(\Omega_D)=\min\left\{\frac{1}{n-1},\frac1{2}-\frac{k(k-n)+n}{n(n-1)}\right\}\!.
  \label{eq:gapD}
\end{align}
The above equations can be simplified by virtue of the assumption $n\ge 4$, with the result
\begin{align}
  \lambda_2(\Omega_D)&=1-\frac{1}{n-1}\,,\\
  \nu(\Omega_D)&=\frac{1}{n-1}\,;
\end{align}
in addition, the second largest eigenvalue has multiplicity $n-1$.
Furthermore, it is straightforward to verify that the kets $\ket{\phi_{ij}}$ in
Eq.~\eqref{eq:2ndEigD} are eigenvectors of $\Omega_D$ with eigenvalue
$1-\frac{1}{n-1}$. The span of all $\ket{\phi_{ij}}$ for $1\leq i<j\leq n$ has
dimension $n-1$, which accounts for the multiplicity $n-1$ of the second largest eigenvalue.
\end{proof}

To briefly summarize, Theorem~2 provides an efficient verification protocol for
any Dicke state $\ket{D_n^k}$.
In real experiments, the experimenter performs the following procedure in each
run of the verification protocol:
\setlist{rightmargin=1cm}
\begin{itemize}
  \item
    The $n$ parties use shared randomness or classical communication to
    randomly choose  two parties, e.g., $i$ and $j$ (Ida and Jim).
  \item
    All the rest $n-2$ parties perform Pauli-$Z$ measurements, then send
    their outcomes to Ida and Jim.
    \begin{itemize}
      \item
	If the outcome~$1$ appears $k$ or $k-2$ times for the $n-2$ Pauli-$Z$
	measurements, then Ida and Jim also perform Pauli-$Z$ measurements.
	\begin{itemize}
	  \item
	    If exactly $k$ of the $n$ Pauli-$Z$ measurements give
	    outcome~$1$, then they announce the result ``pass''; otherwise, they
	    announce the result ``fail''.
	\end{itemize}
      \item
	If the outcome~$1$ appears $k-1$ times for the $n-2$ Pauli-$Z$
	measurements, then Ida and Jim perform Pauli-$X$ measurements.
	\begin{itemize}
	  \item
	    If the two Pauli-$X$ measurements give the same outcome, then they
	    announce the result ``pass''; otherwise, they announce the result
	    ``fail''.
	\end{itemize}
      \item
	If the outcome~$1$ appears less than $k-2$ or more than $k$ times for
	the $n-2$ Pauli-$Z$ measurements, then Ida and Jim announce the result
	``fail''.
    \end{itemize}
\end{itemize}

\section{Proof of the saturation of the bound in Eq.~\eqref{eq:efficiencyD}}
Similar to the case of $W$ states, we can prove the saturation of the bound in
Eq.~\eqref{eq:efficiencyD} when $n\geq 4$, as shown in the following lemma.
\begin{lemma}\label{lem:Defficiency}
    When $n\geq 4$, $\nu(\tilde\Omega_D)=\frac{1}{2}\nu(\Omega_D)$.
\end{lemma}

\begin{proof}
  In the main text, we have already proved that
  $\nu(\tilde\Omega_D)\ge\frac{1}{2}\nu(\Omega_D)$. Hence, to prove the saturation of the bound,
  we just need to show that
  $\nu(\tilde\Omega_D)\le\frac{1}{2}\nu(\Omega_D)$. Note that
  $\nu(\tilde\Omega_D)$ can be written as
  \begin{equation}
    \nu(\tilde{\Omega}_D):=1-\max_{\braket{\phi}{D_n^k}=0}\bra{\phi}\tilde{\Omega}_D\ket{\phi}\,.
    \label{eq:vD}
  \end{equation}
  By taking $\ket{\phi}$ to be $\ket{\phi_{ij}}$ defined in
  Eq.~\eqref{eq:2ndEigD}, which are orthogonal to $\ket{D_n^k}$, we get an
  upper bound of $\nu(\tilde{\Omega}_D)$, i.e.,
  \begin{equation}
    \nu(\tilde{\Omega}_D)\le 1-\bra{\phi_{ij}}\tilde{\Omega}_D\ket{\phi_{ij}}
    =\frac{1}{2(n-1)}=\frac{1}{2}\nu(\Omega_D)\,,
    \label{eq:vDupper}
  \end{equation}
which confirms the lemma.
\end{proof}

\section{Adaptive verification of the Dicke state $\ket{D_4^2}$}
The state $\ket{D_4^2}$ has $k=2$ excitations,
and the verification operator $\Omega_{D_4^2}$ of the adaptive protocol in
Theorem~\ref{theo:D} takes on the form
\begin{widetext}
\begin{align}
  \Omega_{D_4^2}=\frac1{6}\Bigl[&Z_4^{-}Z_3^{+}(XX)_{2,1}^{+}+Z_4^{+}Z_3^{-}(XX)_{2,1}^{+}
  +Z_4^{-}Z_3^{-}Z_2^{+}Z_1^{+}+Z_4^{+}Z_3^{+}Z_2^{-}Z_1^{-}\nonumber\\
            +&Z_4^{-}Z_2^{+}(XX)_{3,1}^{+}+Z_4^{+}Z_2^{-}(XX)_{3,1}^{+}
  +Z_4^{-}Z_2^{-}Z_3^{+}Z_1^{+}+Z_4^{+}Z_2^{+}Z_3^{-}Z_1^{-}\nonumber\\
            +&Z_3^{-}Z_2^{+}(XX)_{4,1}^{+}+Z_3^{+}Z_2^{-}(XX)_{4,1}^{+}
  +Z_3^{-}Z_2^{-}Z_4^{+}Z_1^{+}+Z_3^{+}Z_2^{+}Z_4^{-}Z_1^{-}\nonumber\\
            +&Z_4^{-}Z_1^{+}(XX)_{3,2}^{+}+Z_4^{+}Z_1^{-}(XX)_{3,2}^{+}
  +Z_4^{-}Z_1^{-}Z_3^{+}Z_2^{+}+Z_4^{+}Z_1^{+}Z_3^{-}Z_2^{-}\nonumber\\
            +&Z_3^{-}Z_1^{+}(XX)_{4,2}^{+}+Z_3^{+}Z_1^{-}(XX)_{4,2}^{+}
  +Z_3^{-}Z_1^{-}Z_4^{+}Z_2^{+}+Z_3^{+}Z_1^{+}Z_4^{-}Z_2^{-}\nonumber\\
            +&Z_2^{-}Z_1^{+}(XX)_{4,3}^{+}+Z_2^{+}Z_1^{-}(XX)_{4,3}^{+}
  +Z_2^{-}Z_1^{-}Z_4^{+}Z_3^{+}+Z_2^{+}Z_1^{+}Z_4^{-}Z_3^{-}\Bigr].
\end{align}
\end{widetext}
The second largest eigenvalue of $\Omega_{D_4^2}$ is $\lambda_2(\Omega_{D_4^2})=\frac{2}{3}$,
and the spectral gap is $\nu(\Omega_{D_4^2})=\frac1{3}$. Therefore, the number of tests
required to verify  $\ket{D_4^2}$ within infidelity $\epsilon$ and confidence $1-\delta$
is $N\approx3\epsilon^{-1}\ln\delta^{-1}$.
Simulation results on the verification of $\ket{D_4^2}$ can be found in Appendix~C.


%

\end{document}